\documentclass[conference,letterpaper]{IEEEtran}  % Comment this line out
\pdfoutput=1                                                         % if you need a4paper
\IEEEoverridecommandlockouts                              % This command is only
\usepackage{epsfig} % for postscript graphics files
\usepackage{cite} % assumes new font selection scheme installed
\usepackage{amsmath} % assumes amsmath package installed
\usepackage{amssymb}  % assumes amsmath package installed

\newtheorem{theorem}{Theorem}%[section]
\newtheorem{lemma}{Lemma}
\newtheorem{definition}{Definition}

\newtheorem{remark}{Remark}

\newtheorem{assumption}{Assumption}

\newcommand{\sr}{\stackrel}

\newcommand{\rar}{\rightarrow}

\newcommand{\tri}{\sr{\triangle}{=}}

\newcommand{\noi}{\noindent}

\newcommand{\be}{\begin{equation}}
\newcommand{\ee}{\end{equation}}
\newcommand{\bea}{\begin{eqnarray}}
\newcommand{\eea}{\end{eqnarray}}
\newcommand{\bes}{\begin{eqnarray*}}
\newcommand{\ees}{\end{eqnarray*}}

\newcommand{\bfi}{\begin{figure}}
\newcommand{\bfit}{\begin{figure}[t]}
\newcommand{\bfib}{\begin{figure}[b]}
\newcommand{\bfih}{\begin{figure}[h]}
\newcommand{\bfip}{\begin{figure}[p]}
\newcommand{\efi}{\end{figure}}

\newcommand{\bi}{\begin{itemize}}
\newcommand{\ei}{\end{itemize}}
\newcommand{\ben}{\begin{enumerate}}
\newcommand{\een}{\end{enumerate}}

\title{ \LARGE\bf
On the relation of nonanticipative rate distortion function and filtering theory}

\author{Charalambos D. Charalambous and Photios A. Stavrou% <-this % stops a space
\thanks{*This work was financially supported by a medium size University of Cyprus grant entitled ``DIMITRIS''.}% <-this % stops a space
\thanks{The authors are with the Department of Electrical and Computer Engineering (ECE), University of Cyprus, Nicosia, CYPRUS
        {\tt\small chadcha@ucy.ac.cy, stavrou.fotios@ucy.ac.cy}.}%
}

\begin{document}

\maketitle
\thispagestyle{empty}
\pagestyle{empty}

%%%%%%%%%%%%%%%%%%%%%%%%%%%%%%%%%%%%%%%%%%%%%%%%%%%%%%%%%%%%%%%%%%%%%%%%%%%%%%%%
\begin{abstract}

In this paper the relation between nonanticipative rate distortion function (RDF) and Bayesian filtering theory is investigated using the topology of weak convergence of probability measures on Polish spaces. The relation is established via an optimization on the space of conditional distributions of the so-called directed information subject to fidelity constraints. Existence of the optimal reproduction distribution of the nonanticipative RDF is shown, while the optimal nonanticipative reproduction conditional distribution for stationary processes is derived in closed form. The realization procedure of nonanticipative RDF which is equivalent to joint-source channel matching for symbol-by-symbol transmission is described, while an example is introduced to illustrate the concepts.

\end{abstract}

%%%%%%%%%%%%%%%%%%%%%%%%%%%%%%%%%%%%%%%%%%%%%%%%%%%%%%%%%%%%%%%%%%%%%%%%%%%%%%%%
\section{Introduction}

This paper is concerned with the abstract formulation of nonanticipative rate distortion function (RDF) on Polish spaces (complete separable metric spaces) and its relation to filtering theory. In the past, rate distortion (or distortion rate) functions and filtering theory have evolved independently. Specifically, classical RDF addresses the problem of reproduction of a process subject to a fidelity criterion without much emphasis on the realization of the reproduction conditional distribution via nonanticipative operations. On the other hand, filtering theory is developed by imposing real-time realizability on estimators with respect to measurement data.\\ %Specifically, least-squares filtering theory deals with the characterization of the conditional distribution of the unobserved process given the measurement data, via a stochastic differential equation which depends on the observation data \cite{elliott-aggoun-moore1995} via nonanticipative operations.\\
%\noi  Although, both reliable communication and filtering (state estimation for control) are concerned with the reproduction of processes, the main underlying assumptions characterizing them are different. Moreover, .\\ 
\noi Historically, the work of R. Bucy  \cite{bucy} appears to be the first to consider the direct relation between distortion rate function and filtering. The work of A. K. Gorbunov and M. S. Pinsker \cite{gorbunov91} on $\epsilon$-entropy defined via a nonanticipative constraint on the reproduction distribution of the RDF, although not directly related to the realizability question pursued by  Bucy, computes the nonanticipative RDF for stationary Gaussian processes via power spectral densities. \\
\noi The objective of this paper is to investigate the connection between nonanticipative RDF and filtering theory for general distortion functions and random processes on abstract Polish spaces using the topology of weak convergence.\\ %The connection is established  via optimization of directed information \cite{massey90} over the space of conditional distributions which satisfy an average distortion constraint. \\
The main results discussed in this paper are the following.\\
{\bf(1)} Existence of optimal reproduction distribution minimizing directed information using the topology of weak convergence of probability measures on Polish spaces;\\
{\bf(2)} Closed form expression of the optimal reproduction conditional distribution for stationary processes;\\
{\bf(3)} Realization procedure of the filter;\\
{\bf(4)} Example to demonstrate the realization of the filter;\\
{\bf(5)} Connection between nonanticipative RDF and joint source-channel coding of symbol-by-symbol transmission \cite{gastpar2003}.\\
\noi {\it Motivation.} This work is motivated by applications in which estimators are desired to have specific accuracy, by control over limited rate communication channel applications \cite{tatikonda-mitter2004b,nair-evans2004}, and by the desire to provide necessary conditions for symbol-by-symbol or uncoded transmission \cite{gastpar2003} for sources with memory without anticipation.  
\par First, we give a brief high level discussion on nonanticipative RDF and filtering theory, and discuss their connection. Consider a discrete-time process $X^n\tri\{X_0,X_1,\ldots,X_n\}\in{\cal X}_{0,n} \tri \times_{i=0}^n{\cal X}_i$, and its reproduction $Y^n\tri\{Y_0,Y_1,\ldots,Y_n\}\in{\cal Y}_{0,n} \tri \times_{i=0}^n{\cal Y}_i$ where ${\cal X}_i$ and ${\cal Y}_i$ are Polish spaces.\\
\noi{\it Bayesian Estimation Theory.} In classical filtering, one is given a mathematical model that generates the process $X^n$, $\{P_{X_i|X^{i-1}}(dx_i|x^{i-1}):i=0,1,\ldots,n\}$, a mathematical model that generates observed data obtained from sensors, say, $Z^n$, $\{P_{Z_i|Z^{i-1},X^i}$ $(dz_i|z^{i-1},x^i):i=0,1,\ldots,n\}$, while $Y^n$ are the causal estimates of some function of the process $X^n$ based on the observed data $Z^n$.
The classical Kalman Filter is a well-known example, where $\widehat{X}_i =\mathbb{E}[X_i | Z^{i-1}],~i=0,1,\ldots,n$, is the conditional mean which minimizes the average least-squares estimation error. Fig. 1 is the block diagram of the filtering problem.	
\begin{figure}[ht]
\centering
\includegraphics[scale=0.5]{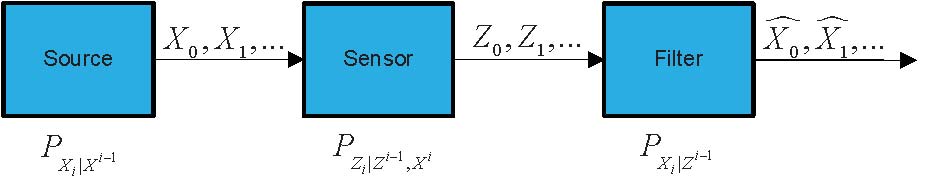}
\caption{Filtering problem.}
\label{filtering}
\end{figure}

\noi{\it Nonanticipative Rate Distortion Theory and Estimation.} In nonanticipative rate distortion theory one is given a  distribution for the process $X^n$, which induces $\{P_{X_i|X^{i-1}}(dx_i|x^{i-1}):~i=0,1,\ldots,n\}$, and determines the nonanticipative reproduction conditional distribution $\{P_{Y_i|Y^{i-1},X^i}(dy_i|y^{i-1},x^i):~i=0,1,\ldots,n\}$ which minimizes the directed information from $X^n$ to $Y^n$ subject to distortion or fidelity constraint. The filter $\{Y_i:~i=0,1,\ldots,n\}$ of $\{X_i:~i=0,1,\ldots,n\}$ is found by realizing the optimal reproduction distribution $\{P_{Y_i|X^{i-1},X^i}(dy_i|y^{i-1},x^i):~i=0,1,\ldots,n\}$ via a cascade of sub-systems as shown in Fig. 2. Thus, in nonanticipative rate distortion theory the observation or mapping from $\{X_i:~i=0,1,\ldots,n\}$ to $\{Z_i:~i=0,1,\ldots,n\}$ is part of the realization procedure, while in filtering theory, this mapping is given \'a priori. %Indeed, this is the main difference between Bayesian estimation theory and nonanticipative RDF for the purpose of estimation.
\begin{figure}[ht]
\centering
\includegraphics[scale=0.5]{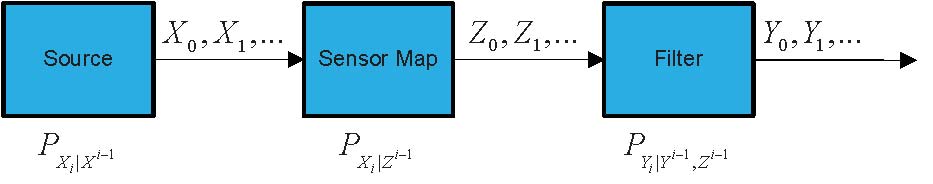}
\caption{Filtering via nonanticipative rate distortion function.}
\label{filtering_and_causal}
\end{figure}

\noi The precise problem formulation necessitates  the definitions of  distortion function or fidelity, and directed information.\\
The distortion function or fidelity constraint \cite{berger} between $x^n$ and its reproduction $y^n$, is a measurable function  $d_{0,n} : {\cal X}_{0,n} \times {\cal Y}_{0,n} \rar [0, \infty]$ defined by
\begin{align*}
d_{0,n}(x^n,y^n)\tri\frac{1}{n+1}\sum^n_{i=0}\rho_{0,i}(x^i,y^i).
\end{align*}
%where $\rho_{0,i}: {\cal X}_{0,i}  \times {\cal Y}_{0,i}\rightarrow [0, \infty]$, is a sequence of ${\cal B}({\cal X}_{0,i}) \times {\cal B }( {\cal Y}_{0,i})$-measurable distortion functions.\\
Directed information  from a sequence of Random Variables (RV's) $X^n\tri\{X_0,X_1,\ldots,X_n\}\in{\cal X}_{0,n}\tri\times_{i=0}^n{\cal X}_i$, to another sequence $Y^n\tri\{Y_0,Y_1,\ldots,Y_n\}\in{\cal Y}_{0,n}\tri\times_{i=0}^n{\cal Y}_i$ is often defined via \cite{massey90,charalambous-stavrou2012}\footnote[4]{Unless otherwise, integrals with respect to probability distributions are over the spaces on which these are defined.} 
\begin{align}
&I(X^n\rightarrow{Y}^n)\tri\sum_{i=0}^n{I}(X^i;Y_i|Y^{i-1})\nonumber\\
&=\sum_{i=0}^n\int\log\Big(\frac{P_{Y_i|Y^{i-1},X^i}(dy_i|y^{i-1},x^i)}{{P}_{Y_i|Y^{i-1}}(dy_i|y^{i-1})}\Big)P_{X^i,Y^i}(dx^i,dy^i)\nonumber\\%\label{1a}
&\equiv\mathbb{I}_{X^n\rightarrow{Y^n}}(P_{X_i|X^{i-1},Y^{i-1}},P_{Y_i|Y^{i-1},X^i}:~i=0,1,\ldots,n).\nonumber%\label{1b}
\end{align}
\noi In this paper, it is assumed that $\forall~i=0,1,\ldots,n$,
\begin{align}
P_{X_i|X^{i-1},Y^{i-1}}(dx_i|x^{i-1},y^{i-1})=P_{X_i|X^{i-1}}(dx_i|x^{i-1}).\nonumber %\label{2}
\end{align}
The above assumption states that the process $\{X_i:~i=0,1,\ldots,n\}$ is conditionally independent of $Y^{i-1}=y^{i-1}$ given knowledge of $X^{i-1}=x^{i-1}$, and it is implied by the following conditional independence, $P_{Y_i|Y^{i-1},X^{\infty}}$ $(dy_i|y^{i-1},x^{\infty})=P_{Y_i|Y^{i-1},X^i}(dy_i|y^{i-1},x^i)-a.s.,~\forall~i=0,1,\ldots,n$. The last assumption implies that the reproduction of $Y_i$ does not depend on future values $X_{i+1}^{\infty}\tri\{X_{i+1},X_{i+2},\ldots,X_{\infty}\}$.\\
Given a sequence of source distributions $\{{P}_{X_i|X^{i-1}}(\cdot|\cdot):~i=0,1,\ldots,n\}$ and a sequence of reproduction conditional distributions $\{P_{Y_i|Y^{i-1},X^i}(\cdot|\cdot,\cdot):~i=0,1,\ldots,n\}$ define the joint distribution $P_{X^n,Y^n}(dx^n,dy^n)={P}_{X_i|X^{i-1}}(dx_i|x^{i-1})\otimes{P}_{Y_i|Y^{i-1},X^i}(dy_i|y^{i-1},x^i)$. The  nonanticipative RDF is a special case of directed information defined by 
\begin{align}
&I_{P_{X^n}}(X^n\rightarrow{Y^n})\nonumber\\
&=\mathbb{I}_{X^n\rightarrow{Y^n}}(P_{X_i|X^{i-1}},P_{Y_i|Y^{i-1},X^i}:i=0,1,\ldots,n).\nonumber%\label{3}
\end{align}
\noi {\it Nonanticipative RDF.}
The nonanticipative RDF is defined by
\begin{equation}
{R}^{na}_{0,n}(D)\tri \inf_{\substack{P_{Y_i|Y^{i-1},X^i}(\cdot|\cdot,\cdot),\\~i=0,1,\ldots,n:\\
\mathbb{E}\big\{d_{0,n}(X^n,Y^n)\leq{D}\big\}}}I_{P_{X^n}}(X^n\rightarrow{Y^n}).\label{7}
\end{equation}
The definition of the nonanticipative RDF is consistent with \cite{gorbunov-pinsker} in which nonanticipation is defined via the Markov chain (MC) $X_{n+1}^\infty \leftrightarrow X^n \leftrightarrow Y^n$, e.g., $P_{Y^n|X^{\infty}}(dy^n|x^{\infty})=P_{Y^n|X^n}(dy^n|x^n)$. Therefore, by finding the solution of  (\ref{7}), then one can realize it via a channel from which one can construct an optimal filter via nonanticipative operations as in Fig.~\ref{filtering_and_causal}. One can view the sensor map as consisting of an encoder and a channel, thus draw relations to symbol-by-symbol and uncoded transmission in information theory \cite{gastpar2003}.
\par This paper is organized as follows. Section~\ref{abstract} discusses the formulation on abstract spaces. Section~\ref{existence} establishes  existence of optimal minimizing  distribution,  and Section~\ref{necessary} derives the optimal minimizing distribution for stationary processes. Section~\ref{realization1} describes the  realization of nonanticipative RDF, while Section~\ref{example} provides an example. 
%%%%%%%%%%%%%%%%%%%%%%%%%%%%%%%%%%%%%%%%%%%%%%%%%%%%%%%%%%%%%%%%%%%%%%%%%%%%%%%%%%%%%%%%%
\section{Abstract Formulation}\label{abstract}

The source and reproduction alphabets are sequences of Polish spaces \cite{dupuis-ellis97}. Probability distributions on any measurable space  $( {\cal Z}, {\cal B}({\cal Z}))$ are denoted by ${\cal M}_1({\cal Z})$. For $({\cal X}, {\cal B}({\cal X})), ({\cal Y}, {\cal B}({\cal Y}))$  measurable spaces, the set of conditional distributions  $P_{Y|X}(\cdot|X=x)$ is denoted by ${\cal Q}({\cal Y};{\cal X})$ and these are equivalent to stochastic kernels on $({\cal Y},{\cal B}({\cal Y}))$ given $({\cal X},{\cal B}({\cal X}))$.\\
Given the process distributions $P_{X^n}(dx^n)$ and $\{P_{Y_i|Y^{i-1},X^i}(dy_i|y^{i-1},x^i):~i=0,1,\ldots,n\}$ the following probability distributions are defined.\\
({\bf P1}): The reproduction conditional probability distribution ${\overrightarrow P}_{Y^n|X^n}\in \overrightarrow{\cal Q}({\cal Y}_{0,n};{\cal X}_{0,n})$:
\begin{equation}
{\overrightarrow P}_{Y^n|X^n}(dy^n|x^n) \tri \otimes^n_{i=0}P_{Y_i|Y^{i-1},X^i}(dy_i|y^{i-1},x^i).\nonumber%\label{4}
\end{equation}
({\bf P2}): The joint probability distribution $P_{X^n,Y^n}\in {\cal M}_1({\cal Y}_{0,n}\times {\cal X}_{0, n})$ for $G_{0,n} \in {\cal B}({\cal X}_{0,n})\times{\cal B}({\cal Y}_{0,n})$:
\begin{align}
P_{X^n,Y^n}&(G_{0,n})\tri(P_{X^n} \otimes \overrightarrow{P}_{Y^n|X^n})(G_{0,n})\nonumber\\
&=\int \overrightarrow{P}_{Y^n|X^n}(G_{0,n,x^n}|x^n)\otimes{P}_{X^n}(d{x^n})\nonumber
\end{align}
where $G_{0,n,x^n}$ is the $x^n-$section of $G_{0,n}$ at point ${x^n}$ defined by $G_{0,n,x^n}\tri \{y^n \in {\cal Y}_{0,n}: (x^n, y^n) \in G_{0,n}\}$ and $\otimes$ denotes the convolution.\\
({\bf P3}): The marginal distribution $P_{Y^n}\in {\cal M}_1({\cal Y}_{0,n})$:% corresponding to the kernel $q_{0,n} \in {\cal Q}({\cal Y}_{0, n};{\cal X}_{0, n})$:
\begin{align}
P_{Y^n}&(F_{0,n})\tri P({\cal X}_{0, n} \times F_{0,n}),~F_{0,n} \in {\cal B}({\cal Y}_{0,n})\nonumber\\
&=\int \overrightarrow{P}_{Y^n|X^n}(F_{0,n}|x^n) P_{X^n}(d{x^n}).\nonumber
\end{align}
Define
\begin{align*}
&\overrightarrow{\cal Q}({\cal Y}_{0,n};{\cal X}_{0,n})=\Big\{{\overrightarrow P}_{Y^n|X^n}(dy^n|x^n)\in{\cal Q}({\cal Y}_{0,n};{\cal X}_{0,n}):\\
&{\overrightarrow P}_{Y^n|X^n}(dy^n|x^n) \tri \otimes^n_{i=0}P_{Y_i|Y^{i-1},X^i}(dy_i|y^{i-1},x^i)\Big\}.
\end{align*}
\noi Directed information (special case) is defined via the Kullback-Leibler distance:
\begin{align}
&I_{P_{X^n}}(X^n\rightarrow{Y^n})\tri\mathbb{D}(P_{X^n,Y^n}|| P_{X^n}\times{P_{Y^n}})\nonumber\\
&=\mathbb{D}(P_{X^n}\otimes{\overrightarrow{P}}_{Y^n|X^n}||P_{X^n}\times{P}_{Y^n})\nonumber\\
&=\int\log \Big( \frac{d  (P_{X^n} \otimes \overrightarrow{P}_{Y^n|X^n}) }{d ( P_{X^n} \times P_{Y^n} ) }\Big) d(P_{X^n} \otimes\overrightarrow{P}_{Y^n|X^n}) \nonumber\\
%& = \int \log \Big( \frac{\overrightarrow{P}_{Y^n|X^n}(dy^n|x^n)}{P_{Y^n}(dy^n)} \Big)\overrightarrow{P}_{Y^n|X^n}(dy^n|x^n)\otimes{P}_{X^n}(dx^n)\nonumber\\
%&=\int\mathbb{D}(P_{Y^n|X^n}(\cdot|x^n)|| P_{Y^n}(\cdot))P_{X^n}(dx^n)\nonumber\\
&\equiv \mathbb{I}_{X^n\rightarrow{Y^n}}(P_{X^n},\overrightarrow{P}_{Y^n|X^n}).  \label{re3}
 \end{align}
Note that (\ref{re3}) states that directed information is expressed as a functional of $\{P_{X^n},\overrightarrow{P}_{Y^n|X^n}\}$. \\
Next, the definition of nonanticipative RDF is given.
\begin{definition}\label{def1}
$(${\bf Nonanticipative RDF}$)$
Suppose $d_{0,n}\tri\sum^n_{i=0}\rho_{0,i}(x^i,y^i)$ is measurable, and let $\overrightarrow{\cal Q}_{0,n}(D)$ (assuming is non-empty) denotes the fidelity set 
\begin{align}
&\overrightarrow{\cal Q}_{0,n}(D)\tri\big\{\overrightarrow{P}_{Y^n|X^n} \in \overrightarrow{\cal Q}({\cal Y}_{0,n};{\cal X}_{0,n}) :~\ell_{d_{0,n}}(\overrightarrow{P}_{Y^n|X^n})\nonumber \\
&\tri\int d_{0,n}(x^n,y^n) 
\overrightarrow{P}_{Y^n|X^n}(dy^n|x^n)\otimes{P}_{X^n}(dx^n)\leq D\big\}\label{eq2}
\end{align}
where $D\geq0$. The nonanticipative RDF is defined by
\begin{align}
{R}^{na}_{0,n}(D) \tri  \inf_{{\overrightarrow{P}_{Y^n|X^n}\in \overrightarrow{\cal Q}_{0,n}(D)}}{\mathbb I}_{X^n\rightarrow{Y^n}}({P}_{X^n},\overrightarrow{P}_{Y^n|X^n}).\label{ex12}
\end{align}
\end{definition}
Clearly, ${R}^{na}_{0,n}(D)$ is characterized by minimizing $\mathbb{I}_{X^n\rightarrow{Y^n}}({P}_{X^n},\overrightarrow{P}_{Y^n|X^n})$ over $\overrightarrow{\cal Q}_{0,n}(D)$.

\section{Existence of Reproduction Distribution}\label{existence}

\par In this section, the existence of the minimizing $(n+1)$-fold convolution of conditional distributions in (\ref{ex12}) is established  by using the topology of weak convergence of probability measures on Polish spaces. First, we state some properties derived in \cite{charalambous-stavrou2012}.
\begin{theorem}\label{convexity_properties}\cite{charalambous-stavrou2012}
Let $\{{\cal X}_n:~n\in\mathbb{N}\}$ and  $\{{\cal Y}_n:~n\in\mathbb{N}\}$ be Polish spaces. Then\\
{\bf(1)} The set $\overrightarrow{\cal Q}({\cal Y}_{0,n};{\cal X}_{0,n})$ is convex.\\
{\bf(2)} ${\mathbb I}_{X^n\rightarrow{Y^n}}({P}_{X^n},\overrightarrow{P}_{Y^n|X^n})$ is a convex functional of $\overrightarrow{P}_{Y^n|X^n}\in\overrightarrow{\cal Q}({\cal Y}_{0,n};{\cal X}_{0,n})$ for a fixed $P_{X^n}\in{\cal M}_1({\cal X}_{0,n})$.\\
{\bf(3)} The set $\overrightarrow{\cal Q}_{0,n}(D)$ is convex.
\end{theorem}
 Let $BC({\cal Y}_{0,n})$ denotes the set of bounded continuous real-valued functions on ${\cal Y}_{0,n}$. We need the following.   
\begin{assumption}\label{conditions-existence}
The following conditions are assumed throughout the paper.\\
{\bf(A1)} ${\cal Y}_{0,n}$ is a compact Polish space, ${\cal X}_{0,n}$ is a Polish space;\\
{\bf(A2)} for all $h(\cdot){\in}BC({\cal Y}_{0,n})$, the function mapping $(x^{n},y^{n-1})\in{\cal X}_{0,n}\times{\cal Y}_{0,n-1}\mapsto\int_{{\cal Y}_n}h(y)P_{Y|Y^{n-1},X^n}(dy|y^{n-1},x^n)\in\mathbb{R}$ 
is continuous jointly in the variables $(x^{n},y^{n-1})\in{\cal X}_{0,n}\times{\cal Y}_{0,n-1}$;\\
{\bf(A3)} $d_{0,n}(x^n,\cdot)$ is continuous on ${\cal Y}_{0,n}$;\\
{\bf(A4)} the distortion level $D$ is such that there exist sequence $(x^n,y^{n})\in{\cal X}_{0,n}\times{\cal Y}_{0,n}$ satisfying $d_{0,n}(x^n,y^{n})<D$.
\end{assumption}
Note that since ${\cal Y}_{0,n}$ is assumed to be a compact Polish space, then by \cite{dupuis-ellis97} probability measures on ${\cal Y}_{0,n}$ are weakly compact. Moreover, the following weak compactness result can be obtained.
\begin{lemma}\label{compactness2}
Suppose Assumption~\ref{conditions-existence} {\bf(A1)}, {\bf(A2)} hold. Then\\
{\bf(1)} The set $\overrightarrow{\cal Q}({\cal Y}_{0,n};{\cal X}_{0,n})$ is weakly compact.\\
{\bf(2)} Under the additional conditions {\bf(A3)}, {\bf(A4)}  the set ${\overrightarrow{\cal Q}}_{0,n}(D)$ is a closed subset of $\overrightarrow{\cal Q}({\cal Y}_{0,n};{\cal X}_{0,n})$ (hence compact).
\end{lemma}
\begin{proof}
The derivation is found in \cite{stavrou-charalambous2013c}.
\end{proof} 

The previous results follow from Prohorov's theorem that relates tightness and weak compactness. The next theorem establishes existence of the minimizing reproduction distribution for (\ref{ex12}); it follows from Lemma~\ref{compactness2} and the lower semicontinuity of $\mathbb{I}_{X^n\rightarrow{Y^n}}(P_{X^n},\cdot)$ with respect to $\overrightarrow{P}_{Y^n|X^n}$ \cite{stavrou-charalambous2013c}.
\begin{theorem}$(${\bf Existence}$)$\label{existence_rd}
Suppose the conditions of Lemma~\ref{compactness2} hold. Then ${R}^{na}_{0,n}(D)$ has a minimum.
\end{theorem}
\begin{proof}
The derivation is found in \cite{stavrou-charalambous2013c}.
\end{proof}

%%%%%%%%%%%%%%%%%%%%%%%%%%%%%%%%%%%%%%%%%%%%%%%%%%%%%%%%%%%%%%%%%%%%%%%%%%%%%%%%

\section{Optimal Reproduction of Nonanticipative RDF}\label{necessary}

In this section the form of the optimal reproduction conditional distribution is derived under a stationarity assumption. We introduce the following main assumption.
\begin{assumption}$(${\bf Stationarity}$)$\label{stationarity}
The $(n+1)$-fold convolution conditional distribution $\overrightarrow{P}_{Y^n|X^n}(dy^n|x^n)=\otimes^n_{i=0}P_{Y_i|Y^{i-1},X^i}$ $(dy_i|y^{i-1},x^i)$, is the convolution of stationary conditional distributions.
\end{assumption}
\noi The consequence of Assumption~\ref{stationarity}, which holds for stationary processes  and a single letter distortion function, is that the Gateaux differential of $\mathbb{I}_{X^n\rightarrow{Y^n}}(P_{X^n},\overrightarrow{P}_{Y^n|X^n})$ is done in only one direction $\overrightarrow{P}_{Y^n|X^n}-\overrightarrow{P}_{Y^n|X^n}^0$ via $\overrightarrow{P}_{Y^n|X^n}^{\epsilon}\tri\overrightarrow{P}_{Y^n|X^n}+\epsilon\big{(}\overrightarrow{P}_{Y^n|X^n}-\overrightarrow{P}_{Y^n|X^n}^0\big{)}$, $\epsilon\in[0,1]$, since under Assumption~\ref{stationarity}, the functionals $\{P_{Y_i|Y^{i-1},X^i}(dy_i|y^{i-1},x^i)\in{\cal Q}({\cal Y}_i;{\cal Y}_{0,i-1}\times{\cal X}_{0,i}):~i=0,1,\ldots,n\}$ are identical.
\begin{theorem} \label{th5}
Suppose Assumption~\ref{stationarity} holds and~${\mathbb I}_{P_{X^n}}(\overrightarrow{P}_{Y^n|X^n}) \tri\mathbb{I}_{X^n\rightarrow{Y^n}}(P_{X^n},\overrightarrow{P}_{Y^n|X^n})$ is well defined for every $\overrightarrow{P}_{Y^n|X^n}\in \overrightarrow{\cal Q}_{0,n}(D)$ possibly taking values from the set $[0,\infty]$. Then  $\overrightarrow{P}_{Y^n|X^n} \rightarrow {\mathbb I}_{P_{X^n}}(\overrightarrow{P}_{Y^n|X^n})$ is Gateaux differentiable at every point in $\overrightarrow{\cal Q}_{0,n}(D)$, and the Gateaux derivative at the  point $\overrightarrow{P}_{Y^n|X^n}^0$ in the direction $\overrightarrow{P}_{Y^n|X^n}-\overrightarrow{P}_{Y^n|X^n}^0$ is given
by
\begin{align}
&\delta{\mathbb I}_{P_{X^n}}(\overrightarrow{P}_{Y^n|X^n}^0,\overrightarrow{P}_{Y^n|X^n}-\overrightarrow{P}_{Y^n|X^n}^0)\nonumber\\
&=\int\log \Bigg(\frac{\overrightarrow{P}_{Y^n|X^n}^0(dy^n|x^n)}{P_{Y^n}^0(dy^n)}\Bigg)\nonumber\\
&\otimes(\overrightarrow{P}_{Y^n|X^n}-\overrightarrow{P}_{Y^n|X^n}^0)(dy^n|x^n) P_{X^n}(dx^n)\nonumber
\end{align}
where $P_{Y^n}^0\in{\cal M}_1({\cal Y}_{0,n})$ is the marginal measure corresponding
to $\overrightarrow{P}_{Y^n|X^n}^0\otimes{P}_{X^n}\in{\cal M}_1({\cal Y}_{0,n}\times{\cal X}_{0,n})$.
\end{theorem}
\begin{proof}
The proof is similar to the one in \cite{farzad06} (although it is more involved).
\end{proof}

\noi The constrained problem defined by (\ref{ex12}) can be reformulated as an unconstrained problem using Lagrange multipliers \cite{stavrou-charalambous2013c} 
\begin{align}
&{R}_{0,n}^{na}(D)= \sup_{s\leq{0}}\inf_{\substack{{\overrightarrow{P}_{Y^n|X^n}}\\
\in \overrightarrow{Q}({\cal Y}_{0,n};{\cal X}_{0,n})}} \Big\{{\mathbb I}_{X^n\rightarrow{Y^n}}(P_{X^n},\overrightarrow{P}_{Y^n|X^n})\nonumber\\
&-s(\ell_{{d}_{0,n}}\big(\overrightarrow{P}_{Y^n|X^n})-D(n+1)\big)\Big\},~s\in(-\infty,0]. \label{ex13}
\end{align}
\noi The above observations yield the following theorem.
\begin{theorem}$(${\bf Optimal Reproduction Distribution}$)$ \label{th6}
Suppose the Assumption~\ref{stationarity} holds and consider $d_{0,n}(x^n,y^n)\tri\sum_{i=0}^n\rho(T^i{x^n},T^i{y^n})$. Then\\
{\bf(1)} The infimum in $(\ref{ex13})$ is attained at  $\overrightarrow{P}^*_{Y^n|X^n} \in\overrightarrow{\cal Q}_{0,n}(D)$ given by\footnote[5]{Due to stationarity assumption $P_{Y_i|Y^{i-1}}(\cdot|\cdot)=P(\cdot|\cdot)$ and ${P}^*_{Y_i|Y^{i-1},X^i}(\cdot|\cdot,\cdot)={P}^*(\cdot|\cdot,\cdot)$}
\begin{align}
&\overrightarrow{P}^*_{Y^n|X^n}(dy^n|x^n)=\otimes_{i=0}^n{P}^*_{Y_i|Y^{i-1},X^i}(dy_i|y^{i-1},x^i)\nonumber\\
&=\otimes_{i=0}^n\frac{e^{s \rho(T^i{x^n},T^i{y^n})}P^*_{Y_i|Y^{i-1}}(dy_i|y^{i-1})}{\int_{{\cal Y}_i} e^{s \rho(T^i{x^n},T^i{y^n})} P^*_{Y_i|Y^{i-1}}(dy_i|y^{i-1})}\label{ex14}
\end{align}
where $s\leq{0}$ and $P^*_{Y_i|Y^{i-1}}(dy_i|y^{i-1})\in {\cal Q}({\cal Y}_i;{\cal Y}_{0,{i-1}})$.\\ 
{\bf(2)} The nonanticipative RDF is given by
\begin{align}
&{R}_{0,n}^{na}(D)=sD(n+1) -\sum_{i=0}^n\int\log \Big( \int_{{\cal Y}_i} e^{s\rho(T^i{x^n},T^i{y^n})} \nonumber\\
&P^*_{Y_i|Y^{i-1}}(dy_i|y^{i-1})\Big){\overrightarrow{P}^*_{Y^{i-1}|X^{i-1}}(dy^{i-1}|x^{i-1})\otimes{P}_{X^i}(dx^i)}.\nonumber%\label{ex15}
\end{align}
If ${R}_{0,n}^{na}(D) > 0$ then $ s < 0$  and
\begin{align}
\sum_{i=0}^n\int\rho(T^i{x^n},T^i{y^n})\overrightarrow{P}^*_{Y^{i}|X^{i}}(dy^i|x^i)P_{X^i}(dx^i)=(n+1)D.\nonumber%\label{eq.7}
\end{align}
\end{theorem}
\begin{proof}
The derivation is found in \cite{stavrou-charalambous2013c}.
\end{proof}
\begin{remark}
Note that if the distortion function satisfies $\rho(T^i{x^n},T^i{y^n})=\rho(x_i,T^i{y^n})$ then for $i=0,1,\ldots,n$
\begin{equation}
{P}^*_{Y_i|Y^{i-1},X^i}(dy_i|y^{i-1},x^i)={P}^*_{Y_i|Y^{i-1},X^i}(dy_i|y^{i-1},x_i)\nonumber
\end{equation}
that is, the reproduction kernel is Markov in $X^n$. 
\end{remark}

%%%%%%%%%%%%%%%%%%%%%%%%%%%%%%%%%%%%%%%%%%%%%%%%%%%%%%%%%%%%%%%%%%%%%%%%%%%%%%%%

\section{Realization of Nonanticipative RDF}\label{realization1}

The realization of the nonanticipative RDF (optimal reproduction conditional distribution) is equivalent to the sensor mapping as shown in Fig.~\ref{filtering_and_causal} which produces the auxiliary random process $\{Z_i:~i\in\mathbb{N}\}$ which is used for filtering. This is equivalent to identifying a communication channel, an encoder and a decoder such that the reproduction from the sequence $X^n$ to the sequence $Y^n$ matches the nonanticipative rate distortion minimizing reproduction kernel. Fig.~\ref{realization_figure} illustrates the cascade  sub-systems that realize the nonanticipative RDF. 
\begin{definition}$(${\bf Realization}$)$\label{realization}
Given a source $\{P_{X_i|X^{i-1}}(dx_i|x^{i-1}):i=0,\ldots,n\}$,  a channel $\{P_{B_i|B^{i-1},A^{i}}(db_i|b^{i-1},a^i):i=0,\ldots,n\}$ is a realization of the optimal reproduction  distribution (\ref{ex14}) if there exists a pre-channel encoder $\{P_{A_i|A^{i-1},B^{i-1},X^i}$ $(da_i|a^{i-1},b^{i-1},x^i):i=0,\ldots,n\}$ and a post-channel decoder $\{P_{Y_i|Y^{i-1},B^i}(dy_i|y^{i-1},b^i):i=0,\ldots,n\}$ such that
\begin{align}
 {\overrightarrow {P}}_{Y^{n}|X^{n}}^*(dy^n|x^n)&\tri\otimes_{i=0}^n P^*_{Y_i|Y^{i-1},X^i}(dy_i|y^{i-1},x^i)\nonumber\\
&=\otimes_{i=0}^n P_{Y_i|Y^{i-1},X^i}(dy_i|y^{i-1},x^i)\label{equation9}
\end{align}
where (\ref{equation9}) is generated from the joint distribution
\begin{align}
&P_{X^n, A^n, B^n, Y^n}(dx^n,da^n,db^n,dy^n) \nonumber \\
&=\otimes_{i=0}^n P_{Y_i|Y^{i-1},B^i}(dy_i|y^{i-1},b^i)\nonumber\\
&\quad\otimes P_{B_i|B^{i-1},A^{i}}
(db_i|b^{i-1},a^i)\nonumber\\
&\quad\otimes P_{A_i|A^{i-1},B^{i-1},X^i}(da_i|a^{i-1},b^{i-1},x^i)\nonumber\\
&\quad\otimes P_{X_i|X^{i-1}}(dx_i|x^{i-1}).\nonumber
\end{align}
The filter is given by $\{P_{X_i|B^{i-1}}(dx_i|b^{i-1}):i=0,\ldots,n\}$ or by $\{P_{X_i|Y^{i-1}}(dx_i|y^{i-1}):i=0,\ldots,n\}$.
\end{definition}
\begin{figure}[ht]
\centering
\includegraphics[scale=0.5]{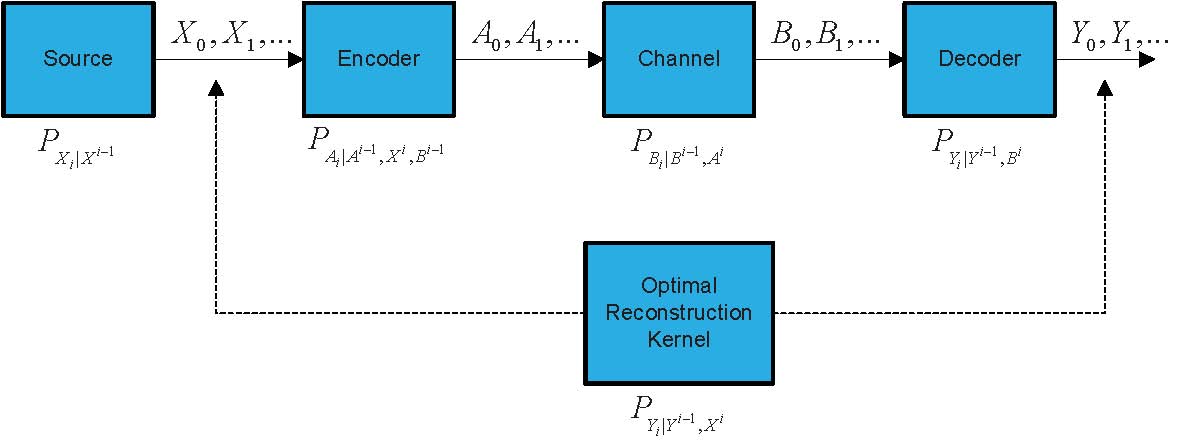}
\caption{Realizable nonanticipative rate distortion function.}
\label{realization_figure}
\end{figure}
\noi Clearly, $\{B_i:~i=0,1,\ldots,n\}$ is an auxiliary random process which is needed to obtain the filter $\{P_{X_i|B^{i-1}}(dx_i|b^{i-1}):i=0,\ldots,n\}$. If we further ensure that there exists $(D,P)$ such that $R^{na}(D)\tri\lim_{n\rightarrow\infty}\frac{1}{n+1}R^{na}_{0,n}(D)=C(P)$, where $C(P)$ is the capacity of the channel with power level $P$, then the realization of Fig.~\ref{realization_figure} is equivalent to symbol-by-symbol transmission in which the source is matched to the channel, e.g., real-time transmission of information.
%%%%%%%%%%%%%%%%%%%%%%%%%%%%%%%%%%%%%%%%%%%%%%%%%%%%%%%%%%%%%%%%%%%%%%%%%%%%%%%%%%%%%%%%%%

\section{Example}\label{example}

\par Consider the following discrete-time partially observed linear Gauss-Markov system described by
\begin{eqnarray}
\left\{ \begin{array}{ll} X_{t+1}=AX_t+BW_t,~X_0=X\in\mathbb{R}^n,~t\in\mathbb{N}\\
Y_t=CX_t+GV_t,~t\in\mathbb{N} \end{array} \right.\label{equation51}
\end{eqnarray}
\noi where $X_t\in\mathbb{R}^m$ is the state (unobserved) process of information source (plant), and $Y_t\in\mathbb{R}^p$ is the partially measurement (observed) process. Assume that ($C,A$) is detectable and ($A,\sqrt{BB^{tr}}$) is stabilizable, ($G\neq0$). The state and observation noises $\{(W_t,V_t):t\in\mathbb{N}\}$, $W_t\in\mathbb{R}^k$ and $V_t\in\mathbb{R}^p$, are Gaussian IID processes with zero mean and identity covariances are mutually independent, and independent of the Gaussian RV $X_0$, with parameters $N(\bar{x}_0,\bar{V}_0)$.\\
\noi The objective is to reconstruct $\{Y_t:~t\in\mathbb{N}\}$ from $\{\tilde{Y}_t:~t\in\mathbb{N}\}$ using single letter distortion. 
%\begin{eqnarray*}
%d_{0,n}(y^n,\tilde{y}^n)\tri\frac{1}{n+1}\sum_{t=0}^n||y_t-\tilde{y}_t||^2.
%\end{eqnarray*}
First, we compute 
\begin{eqnarray*}
R_{0,n}^{na}(D)=\inf_{\overrightarrow{P}_{\tilde{Y}^n|Y^n}\in\overrightarrow{\cal Q}_{0,n}(D)}\frac{1}{n+1}\mathbb{I}_{X^n\rightarrow{Y^n}}(P_{Y^n},\overrightarrow{P}_{\tilde{Y}^n|Y^n})
\end{eqnarray*}
and then realize the optimal reproduction distribution. According to Theorem~\ref{th6}, the optimal reproduction is given by 
\begin{align}
\overrightarrow{P}^*_{\tilde{Y}^n|Y^n}(d\tilde{y}^n|y^n)=\otimes_{t=0}^n\frac{e^{s||\tilde{y}_t-y_t||^2}P_{\tilde{Y}_t|\tilde{Y}^{t-1}}(d\tilde{y}_t|\tilde{y}^{t-1})}{\int_{{\cal Y}_t}e^{s||\tilde{y}_t-y_t||^2}P_{\tilde{Y}_t|\tilde{Y}^{t-1}}(d\tilde{y}_t|\tilde{y}^{t-1})}\label{eq.9}
\end{align}
where $s\leq{0}$.
Hence, from (\ref{eq.9}) it follows that $P_{\tilde{Y}_t|\tilde{Y}^{t-1},Y^t}=P_{\tilde{Y}_t|\tilde{Y}^{t-1},Y_t}(d\tilde{y}_t|\tilde{y}^{t-1},y_t)-a.a.$, that is, the reproduction is Markov with respect to the process $\{Y_t:~t\in\mathbb{N}\}$, and $\{(X_t,{Y}_t):~t\in\mathbb{N}\}$ is jointly Gaussian, hence it follows that $P_{\tilde{Y}_t|\tilde{Y}^{t-1},Y_t}(\cdot|\tilde{y}^{t-1},y_t)$ is Gaussian. Hence, it has the general form
\begin{eqnarray}
\tilde{Y}_t=\bar{A}Y_t+\bar{B}\tilde{Y}^{t-1}+\bar{Z}_t,~t\in\mathbb{N}\label{eq.10}
\end{eqnarray}
where $\bar{A}_t\in\mathbb{R}^{p\times{p}}$, $\bar{B}_t\in\mathbb{R}^{p\times{t}p}$, and $\{\bar{Z}_t:~t\in\mathbb{N}\}$ is an independent sequence of Gaussian vectors.\\
The nonanticipative RDF is given by \cite{stavrou-charalambous2013c} 
\begin{align}
R_{0,n}^{na}(D)=\frac{1}{n+1}\sum_{t=0}^n\sum_{i=1}^p\log\Big{(}\frac{\lambda_{t,i}}{\delta_{t,i}}\Big{)}\label{equation11}
\end{align}
where $\{\xi_t:~t\in\mathbb{N}\}$ are such that 
\begin{eqnarray}
\delta_{t,i} \tri \left\{ \begin{array}{ll} \xi_t & \mbox{if} \quad \xi_t\leq\lambda_{t,i} \\
\lambda_{t,i} &  \mbox{if}\quad\xi_t>\lambda_{t,i} \end{array} \right.,~t\in\mathbb{N},~i=1,\ldots,p\nonumber
\end{eqnarray}
and $\{\xi_t:~t\in\mathbb{N}\}$ satisfies $\sum_{i=1}^p\delta_{t,i}=D$. Define $\Delta_t\tri{diag}\{\delta_{t,i},\ldots,\delta_{t,p}\}$.\\
We realize (\ref{eq.10}) and (\ref{equation11}) via a scalar additive Gaussian noise (AGN) channel with feedback defined by 
\begin{eqnarray}
B_t=A_t+Z_t,~Var(Z_t)=Q,~t\in\mathbb{N}\label{eq.11}
\end{eqnarray}
where the encoder is a mapping $A_t=\Phi_t(Y_t,\tilde{Y}^{t-1})$ with power $P_t\tri{E}\{(A_t)^2\}$. Hence, the capacity of (\ref{eq.11}) is $C(P)\tri\lim_{n\rightarrow\infty}\frac{1}{n+1}I(A^n\rightarrow{B}^n)=\lim_{n\rightarrow\infty}\frac{1}{2}\frac{1}{n+1}\sum_{t=0}^n\log\big(1+E\{(A_t)^2\}Var(Z_t)^{-1}\big)=\frac{1}{2}\log(1+\frac{P}{Q})$. \\
{\it Realization of the nonanticipative RDF.} The realization is based on the block diagram of Fig.~\ref{discrete_time_communication_system}. The encoder $\Phi_t(\cdot,\cdot)$ consists of a pre-encoder which produces the Gaussian innovation process $\{K_t:~t\in\mathbb{N}\}$, defined by 
\begin{eqnarray}
K_t\tri{Y}_t-E\Big{\{}Y_t|\sigma\{\tilde{Y}^{t-1}\}\Big{\}},~t\in\mathbb{N}\label{equation52}
\end{eqnarray}
whose covariance is defined by $\Lambda_t\tri{E}\{K_tK_t^{tr}\}$. The decoder consists of a pre-decoder $\{\tilde{K}_t:~t\in\mathbb{N}\}$ which is defined by 
\begin{eqnarray}
\tilde{K}_t\tri\tilde{Y}_t-E\Big{\{}Y_t|\sigma\{\tilde{Y}^{t-1}\}\Big{\}},~t\in\mathbb{N}.\label{eq.12}
\end{eqnarray}
\noi Let $\{E_t:~t\in\mathbb{N}\}$ be the unitary matrix such that
\begin{eqnarray}
E_t\Lambda_t{E}_t^{tr}=diag\{\lambda_{t,1},\ldots\lambda_{t,p}\},~t\in\mathbb{N}.\label{equation53}
\end{eqnarray}
\noi Define $\Gamma_t\tri{E}_tK_t$ and let $\{\tilde{\Gamma}_t:~t\in\mathbb{N}^{n}\}$ denote its reproduction.\\ 
Thus, the pre-encoder can be further scalled by $\Gamma_t=E_tK_t$, and $\Gamma_t$ is compressed by $A_t={\cal A}_t\Gamma_t$ and sent through the AGN channel with feedback, after which  the received signal is decompressed by $\tilde{\Gamma}_t={\cal B}_tB_t$ in the pre-decoder. By the knowledge of the channel output at the decoder, the mean square estimator $\hat{X}_t$ is generated at the decoder (and encoder because $\hat{X}_t\tri{E}\big{\{}X_t|\sigma\{\tilde{Y}_{t-1}\}\big{\}}$). The complete design is illustrated in Fig.~\ref{discrete_time_communication_system}.
\begin{figure}[ht]
\centering
\includegraphics[scale=0.5]{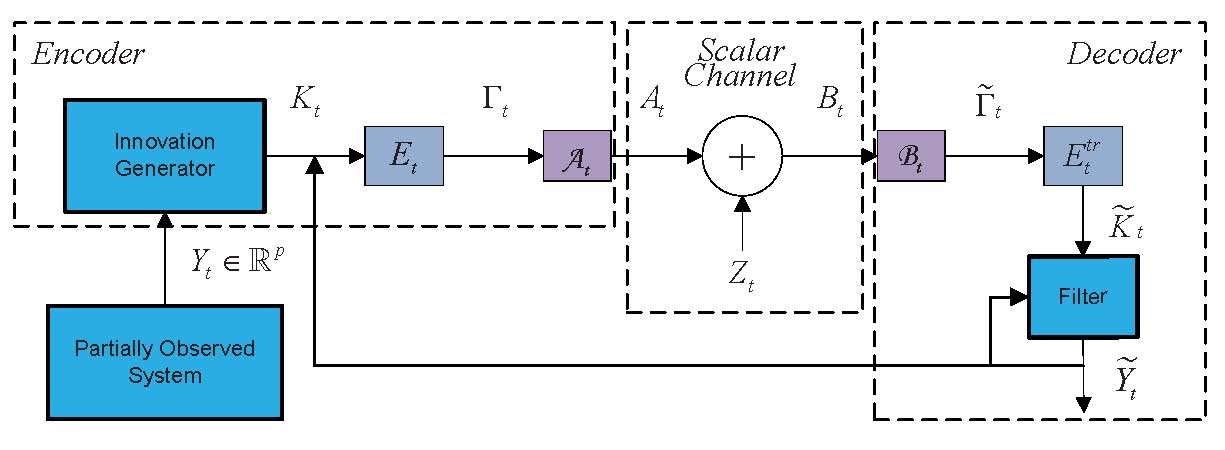}
\caption{Design of the discrete-time communication system with scalar additive Gaussian noise (AGN) channel.}
\label{discrete_time_communication_system}
\end{figure}
\noi We can design $\{({\cal A}_t,{\cal B}_t):~t\in\mathbb{N}\}$ by
\begin{align}
&{\cal A}_t=\Big{[}\sqrt{\frac{\alpha_1{P}_t}{\lambda_{t,1}}},\ldots,\sqrt{\frac{\alpha_p{P}_t}{\lambda_{t,p}}}\Big{]},~t\in\mathbb{N}\nonumber\\%\label{equation55}\\
&{\cal B}_t=\Big{[}\sqrt{\alpha_1{P}_t\lambda_{t,1}},\ldots,\sqrt{\alpha_p{P}_t\lambda_{t,p}}\Big{]}^{tr},~t\in\mathbb{N}\nonumber%\label{equation54}
\end{align}
where $\sum_{i=1}^p{\alpha}_i=1$,~$i=1,\ldots,p$. Note that $H_t\tri{\cal B}_t{\cal A}_t$.\\
{\it Decoder.} From Fig.~\ref{discrete_time_communication_system},
\begin{align*}
\tilde{K}_t&=E_t^{tr}\tilde{\Gamma}_t=E_t^{tr}H_tE_tK_t+E_t^{tr}{\cal B}_tZ_t,~t\in\mathbb{N}.
\end{align*}
The reproduction of $Y_t$ is given by the sum of $\tilde{K}_t$ and $C\hat{X}_t$ as follows.
\begin{align}
\tilde{Y}_t&=E_t^{tr}H_tE_tK_t+E_t^{tr}{\cal B}_tZ_t+C\hat{X}_t,~t\in\mathbb{N}.\nonumber\\%\label{eq.14}
&=E_t^{tr}H_tE_tC(X_t-\hat{X}_t)+C\hat{X}_t\nonumber\\
&\qquad+(E_t^{tr}H_tE_tGV_t+E_t^{tr}{\cal B}_tZ_t)\nonumber
\end{align}
where $\{V_t:~t\in\mathbb{N}\}$ and $\{Z_t:~t\in\mathbb{N}\}$ are independent Gaussian vectors.
\noi The desired distortion is achieved as follows
\begin{align}
&E\Big{\{}(Y_t-\tilde{Y}_t)^{tr}(Y_t-\tilde{Y}_t)\Big{\}}\nonumber\\%=Tr{E}\Big{\{}(K_t-\tilde{K}_t)(K_t-\tilde{K}_t)^{tr}\Big{\}}\nonumber\\
&=Tr\Big{\{}E_t^{tr}\Big{(}(I-H_t)diag(\lambda_{t,1},\ldots,\lambda_{t,p})(1-H_t)^{tr}\nonumber\\
&+({\cal B}_tQ{\cal B}_t^{tr})\Big{)}E_t\Big{\}}=\sum_{i=0}^p\delta_{t,i}=D.\label{equation55}
\end{align}
Thus, from (\ref{equation55}), $\{\delta_{t,i}\}_{i=1}^p$ are eigenvalues of the matrix 
\begin{align*}
T_t\tri(I-H_t)diag(\lambda_{t,1},\ldots,\lambda_{t,p})(1-H_t)^{tr}+({\cal B}_tQ{\cal B}_t^{tr})
\end{align*} 
and we can calculate $\{a_i\}_{i=1}^p$
and $P_t$ in terms of $\{\lambda_{t,i},\delta_{t,i}\}_{i=1}^p$ and $Q$.\\
The decoder is $\tilde{Y}_t=\tilde{K}_t+C\hat{X}_t$, where $\{\hat{X}_t:~t\in\mathbb{N}\}$ is obtained from the modified Kalman filter as follows. 
\begin{align}
\hat{X}_{t+1}&=A\hat{X}_t+A\Sigma_t(E_t^{tr}H_tE_tC)^{tr}M_t^{-1}(\tilde{Y}_t-C\hat{X}_t),\hat{X}_0=\bar{x}_0\nonumber\\
\Sigma_{t+1}&=A\Sigma_tA^{tr}-A\Sigma_t(E_t^{tr}H_tE_tC)^{tr}M_t^{-1}(E_t^{tr}H_tE_tC)\Sigma_tA\nonumber\\
&+BB_t^{tr},~\Sigma_0=\bar{\Sigma}_0\nonumber
\end{align}
where
\begin{align}
M_t&=E_t^{tr}H_tE_tC\Sigma_t(E_t^{tr}H_tE_tC)^{tr}\nonumber\\
&+E_t^{tr}H_tE_tGG^{tr}(E_t^{tr}H_tE_t)^{tr}+E_t^{tr}{\cal B}_tQ{\cal B}_t^{tr}E_t.\nonumber
\end{align}
\noi{\it Infinite Horizon.} As $t\rightarrow\infty$, under the assumption that the linear system is stabilizable and detectable, we have
\begin{align}
\Sigma_{\infty}&=A\Sigma_\infty{A}^{tr}\nonumber\\
&-A\Sigma_{\infty}(E_\infty^{tr}H_\infty{E}_{\infty}C)^{tr}M_{\infty}^{-1}(E_{\infty}^{tr}H_{\infty}E_{\infty}C)\Sigma_{\infty}A\nonumber\\
&+BB_{\infty}^{tr}\nonumber
\end{align}
where
\begin{align}
M_\infty&=E_\infty^{tr}H_\infty{E}_{\infty}C\Sigma_{\infty}(E_{\infty}^{tr}H_{\infty}E_{\infty}C)^{tr}\nonumber\\
&+E_{\infty}^{tr}H_{\infty}E_{\infty}GG^{tr}(E_{\infty}^{tr}H_{\infty}E_{\infty})^{tr}+E_{\infty}^{tr}{\cal B}_{\infty}Q{\cal B}_{\infty}^{tr}E_t\nonumber
\end{align}
and $E_{\infty}$ is the unitary matrix that diagonalizes $\Lambda_{\infty}$ by
\begin{eqnarray}
E_{\infty}\Lambda_{\infty}E_{\infty}^{tr}=diag(\lambda_{\infty,1},\ldots,\lambda_{t,p}).\nonumber
\end{eqnarray}
Also,
\begin{eqnarray}
\delta_{\infty,i} \triangleq \left\{ \begin{array}{ll} \xi_\infty & \mbox{if} \quad \xi_\infty\leq\lambda_{\infty,i} \\
\lambda_{\infty,i} &  \mbox{if}\quad\xi_\infty>\lambda_{\infty,i} \end{array} \right.,~i=1,\ldots,p\nonumber
\end{eqnarray}
satisfying $\sum_{i=1}^p\delta_{\infty,i}=D$. Define $\Delta_{\infty}\tri{diag}(\delta_{\infty,1},\ldots,\delta_{\infty,p})$.\\
Finally, we show matching of the source to the channel.
\begin{align}
&R^{na}(D)\tri\lim_{t\rightarrow\infty}\frac{1}{n+1}R^{na}_{0,n}(D)\nonumber\\
&=\lim_{n\rightarrow\infty}\frac{1}{2}\frac{1}{n+1}\sum_{t=0}^n\sum_{i=1}^p\log\Big(\frac{\lambda_{t,i}}{\delta_{t,i}}\Big)=\frac{1}{2}\sum_{i=1}^p\log\Big(\frac{\lambda_{\infty,i}}{\delta_{\infty,i}}\Big)\nonumber\\
&=\frac{1}{2}\log\frac{|\Lambda_{\infty}|}{|\Delta_{\infty}|}=\frac{1}{2}\log(1+\frac{P}{Q})=C(P).\nonumber
\end{align}
Thus, for a given $(D,P)$, $C(P)=R^{na}(D)$ is the minimum capacity under which there exists a realizable filter for the data reproduction of $\{Y_t:~t\in\mathbb{N}\}$ by $\{\tilde{Y}_t:~t\in\mathbb{N}\}$ ensuring an average distortion equal to $D$. This is precisely the so-called source-channel matching with symbol-by-symbol transmission.

%%%%%%%%%%%%%%%%%%%%%%%%%%%%%%%%%%%%%%%%%%%%%%%%%%%%%%%%%%%%%%%%%%%%%%%%%%%%%%%%%%

%\section{Conclusion}
%%
%In this paper, the solution of the nonanticipative RDF is obtained on abstract spaces using the topology of weak convergence of probability measures. A specific example that realizes the optimal causal filter is discussed.

%%%%%%%%%%%%%%%%%%%%%%%%%%%%%%%%%%%%%%%%%%%%%%%%%%%%%%%%%%%%%%%%%%%%%%%%%%%%%%%%%%%%%%%

\bibliographystyle{IEEEtran}

\bibliography{photis_references_ecc2013}

\end{document}